\documentclass[preprint]{imsart}              

\RequirePackage[OT1]{fontenc}
\RequirePackage{amsthm,amsmath}
\usepackage{graphicx}
\RequirePackage[numbers]{natbib}
\RequirePackage[colorlinks,citecolor=blue,urlcolor=blue]{hyperref}

\arxiv{arXiv:0000.0000}                       

\startlocaldefs
\numberwithin{equation}{section}
\theoremstyle{plain}
\newtheorem{theorem}{Theorem}[section]

\endlocaldefs

\begin{document}

\begin{frontmatter}
  \title{Improving Bayesian estimation \\ of Vaccine Efficacy}

\runtitle{Vaccine Efficacy}

\begin{aug}
\author{\fnms{Mauro} \snm{Gasparini}
\ead[label=e1]{mauro.gasparini@polito.it}}
%
%

\runauthor{Gasparini}

\affiliation{Department of Mathematical Sciences ``G.L. Lagrange'' \\
Politecnico di Torino }
\address{Department of Mathematical Sciences ``G.L. Lagrange'' \\
  Politecnico di Torino \\
\printead{e1}}


\end{aug}

  \begin{abstract}
    A full Bayesian approach to the estimation of Vaccine Efficacy is
    presented, which is an improvement over the currently used exact
    method conditional on the total number of cases.  As an example,
    we reconsider the statistical sections of the BioNTech/Pfizer
    protocol, which in 2020 has led to the first approved
    anti-Covid-19 vaccine.
  \end{abstract}



\end{frontmatter}

\section{Introduction}

The so-called ``exact method conditional on the total number of
cases'' is a Bayesian approach to the estimation of vaccine efficay
(VE) which has been used in the recent pivotal clinical trial of the anti
Covid-19 Vaccine sponsored by Pfizer/BioNTech (\cite{Pfizer}) and has
also been mentioned in the analogous paper about the vaccine sponsored
by Moderna (\cite{Moderna}), the other currently approved mRNA based
vaccine. The name ``exact method conditional on the total number of
cases'' comes from that citation in the latter work.

In addition to the enormous impact of these new therapies on the
lives of billions of people, it should be stressed that these were
some of the first major clinical trials adopting Bayesian methods
for planning and analysis, something many statisticians have been
advocating for quite some time now.

Nonetheless, the exact method conditional on the total number of
cases is only an approximate Bayesian approach, since it makes only
partial use of the full likelihood and of the Bayesian
updating mechanism. In particular, the total number of cases
and the surveillance times
of the vaccinated and of the placebo cohorts are treated
as known parameters instead of observed statistics,
hence the adjective ``conditional'': the method is a partially
Bayesian method conditional on the the total number of cases
and on the surveillance times.

This work contains a more complete full Bayesian approach which takes
as starting point the same assumptions of the conditional method but 
generalizes it by computing the distributions of the total number of
cases and of the surveillance times and by including them in the full
model.

The exact method conditional on the total number of
cases is described in Section 2, while the full Bayesian
model is derived in Section 3.
Section 4 contains a revisitation of the Pfizer/BioNTech
results and a comparison with the full Bayesian approach.

\section{The exact method conditional on the total number of
cases} 

The ``exact method conditional on the total number of cases'' relies
on the mathematical assumption that the infection processes can be
modeled by two overlapping homogeneous Poisson processes: one for
vaccinated participants - with intensity $\lambda_v$ - and an
independent one for the control (not vaccinated) participants - with
intensity $\lambda_c$. The time dimension of the Poisson processes is
called ``surveillance time'' and it is measured in person-years of
follow-up.  It is the sum of all durations participants have been
experiencing in the clinical trial from 7 days (for the BNT162b2 mRNA
vaccine) after the second dose up until the earliest of the following
four endpoints happens: onset of disease, death, loss to follow up or
end of study.  A common measure of comparison between two infection
processes in Epidemiology is the incidence rate ratio IRR=
$\lambda_v/\lambda_c$; based on it, the percentage version of VE is
defined as
\begin{equation}
  100 \times \mathrm{VE} = 100 \times (1- \mathrm{IRR}) =
  100 \times ( 1 - \frac{\lambda_v}{\lambda_c}),
\end{equation}
which can be interpreted as the average
percentage of missed infections
(percentage of not infected vaccinated participants who
would have been infected if not vaccinated).
In order to estimate VE, one can define a likelihood based on
the following statistics:
\begin{itemize}
\item $s_v=$ surveillance time of the vaccine cohort, 
\item $s_c=$ surveillance time of the control cohort, 
\item $x_v+x_c=$ total number of infections,
\item $x_v=$ number of infections in the vaccine cohort.
\end{itemize}
Using standard probability symbolism, the joint density of the corresponding
random variables (indicated in capital letters), which is the dual
way of writing the likelihood, can be expressed as
\begin{multline}
\label{likelihood}
f_{S_V,S_C}(s_V,s_C|\lambda_V,\lambda_C) \times 
f_{X_V+X_C|S_V,S_C}(x_v+x_c|s_V,s_C\lambda_V,\lambda_C) \times \\
f_{X_V|X_V+X_C,S_V,S_C}(x_V|x_v+x_c,s_V,s_C\lambda_V,\lambda_C) 
\end{multline}
i.e. as the chain product of the marginal density of $S_V,S_C$
times the conditional density
of $X_V+X_C$
times the conditional density of  $X_V$ given $x_V+x_C$,
which can be easily proved to be binomial:
\begin{multline}
\label{binomiallikelihood}
f_{X_V|X_V+X_C,S_V,S_C}(x_v|x_v+x_c,s_v,s_c\lambda_v,\lambda_c) = \\
\binom{x_v+x_c}{x_v}
\left(\frac{s_v\lambda_v}{s_v\lambda_v+s_c\lambda_c}\right)^{x_v}
\left(1-\frac{s_v\lambda_v}{s_v\lambda_v+s_c\lambda_c}\right)^{x_c}
\end{multline}
Notice that the probability of infection in this formula is
\begin{equation}
\label{relation}
\theta =  \frac{s_v\lambda_v}{s_v\lambda_v+s_c\lambda_c} =
\frac{s_v(1-\mathrm{VE})}{s_v(1-\mathrm{VE})+s_c}.
\end{equation}
The exact method conditional on the total number of
cases consists of assuming that the first two factors of
the likelihood (\ref{likelihood}) do not depend
on VE, substituing {\em de facto} the binomial expression 
(\ref{binomiallikelihood}) for the full  likelihood (\ref{likelihood}).
To complete the analysis, a conditional Bayesian approach is then followed 
and a conjugate prior Beta(a,b) is given to the parameter $\theta$.
Once the posterior is obtained, the mean a posteriori (MAP) estimate,
Bayesian credible intervals and posterior probabilities
can be computed regarding $\theta$ and, working
equation (\ref{relation})  backward, regarding VE itself. An example
from the Pfizer/BioNTech paper is discussed in Section 4.

\section{An alternative full Bayesian model}

It is not true that the first two factors of
the likelihood (\ref{likelihood}) are independent of VE.
Imagine studies that go on for a long time:
the ratio between the control surveillance time $s_c$
and the vaccinated surveillance time $s_v$
approximates then the ratio of the two mean times
to infection, by the law of large numbers applied once to the numerator
and once to the denominator of the ratio.
Now, the ratio of the two mean times is exactly 1-VE, since the times
to infection are exponentially distributed with mean $1/\lambda_v$
for the vaccinated and $1/\lambda_c$ for the control groups.
Hence, the ratio of the surveillance times does contain
some extra information about VE,
in addition to the numbers of cases in the two groups.
In practice, things are complicated by the fact that the study
must have a finite duration D.

Now, it is not impossible to derive an explicit expression for the
full likelihood (\ref{likelihood}).  Fist, one should notice that
$f_{X_V+X_C|S_V,S_C}(x_v+x_c|s_V,s_C\lambda_V,\lambda_C)$ is
Poisson($s_V\lambda_V+s_C\lambda_C$) by the properties of two
independent overlapping Poisson processes.  Next, the total
surveillance times are the sum of many i.i.d. random variables, each
of them given by the minimum between the time to infection and a
random censoring time.  By the central limit theorem, the surveillance
time of each of the two cohorts is therefore approximated normal as in
the following theorem, where a reasonable specific assumption is made
about the patient recruiting process.

\begin{theorem}
  Under the following assumptions:
  \begin{enumerate}
  \item  an infection process is homogeneous Poisson with
    intensity $\lambda$;
  \item no participant is lost to follow up;
   \item  the recruitment process is uniform
  between beginning of study and a study duration $D$;
  \end{enumerate}
 the mean surveillance time $S/n$, where $n$
  is the number of recruited participants,
  is asymptotically normal with mean
  \begin{equation}
  \label{expectation}
    \mathrm{E}(\min(T,C)) = \frac 1{\lambda}
  (1+ \frac{\exp(-\lambda D)-1}{\lambda D})
      \end{equation}
and variance
\begin{equation}
  \label{variance}
\mathrm{Var}(\min(T,C)) = \frac 1{\lambda^2}
\left( 2 \exp(-\lambda D) + \frac{4\exp(-\lambda D)}{\lambda D}
-(1+ \frac{\exp(-\lambda D)-1}{\lambda D})^2 \right)
  \end{equation}
\end{theorem}

\begin{proof}
  The potential infection time is exponential with rate
  $\lambda$ and can be censored
  by an independent censoring random variable $C$. 
  By assumption 3, $C$ can be written as $D-U$, where
  $U$ is uniform between 0 and the recruitment duration $D$.
Therefore it has itself a uniform distribution between 0 and D.
Next,
  \begin{align*}
  \mathrm{E}(\min(T,C)) &= \int_0^{\infty} \mathrm{P}(\min(T,C)>t) dt \\
  &= \int_0^{\infty} \mathrm{P}(T>t) \mathrm{P}(C>t) dt \\
  &= \int_0^{D} \exp(-\lambda t) \ \frac{D-t}{D} \ dt \\
  &=\frac 1{\lambda} (1+ \frac{\exp(-\lambda D)-1}{\lambda D})
  \end{align*}
    \begin{align*}
  \mathrm{E}(\min(T,C)^2) &= \int_0^{\infty} \mathrm{P}(\min(T,C)^2>t) dt \\
  &= \int_0^{\infty} \mathrm{P}(T>\sqrt{t}) \mathrm{P}(C>\sqrt{t}) dt \\
  &= \int_0^{\infty} \mathrm{P}(T>x) \mathrm{P}(C>x) 2x dx \\
  &=  \frac 1{\lambda^2}
\left(2 \exp(-\lambda D) + \frac{4\exp(-\lambda D)}{\lambda D}\right),
    \end{align*}
    and the variance can be obtained as stated by computing
    $$
    \mathrm{Var}(\min(T,C)) = \mathrm{E}(\min(T,C)^2) -(\mathrm{E}(\min(T,C)))^2
    $$
Finally, the central limit theorem applies.
\end{proof}
Having completed the construction of the likelihood,
to obtain a full Bayesian model only the prior on $(\lambda_v,\lambda_c)$
remains to be decided. A natural proposal is to have
independent gamma priors with hyperparameters $(a_v,b_v)$ and
$(a_c,b_c)$. The scale parameters $b_v$ and $b_c$ should then be chosen
to give $\lambda_v$ and $\lambda_c$ the right order of magnitude.
For example, having a prior guess $\widehat{\lambda_c}$ for the
average number of infected people in the unit time - something
we may estimate based on the natural history of the disease -
we could set $ b_c = a_c/\widehat{\lambda_c}$
and, for the lack of better information, impose $b_v=b_c$.
 Next, the two hyperparameters $a_V$ and $a_C$ can be chosen by gauging
 $VE= 1- \lambda_v/\lambda_c$. For example, noticing that
 \begin{equation}
\label{priorVE}
   \mathrm{E}(\text{VE}) = 1-\mathrm{E}(\lambda_v)\mathrm{E}(1/\lambda_c)
= 1-\frac{a_v}{b_v}\frac{b_c}{a_c-1},
 \end{equation}
then, if $b_c=b_v$, then we could finally set
$$a_c=\frac{1+a_c-\widehat{\text{VE}}}{1-\widehat{\text{VE}}}$$
where $\widehat{\text{VE}}$ is a suitable prior guess for VE,
with $0 \leq VE \leq 1$ (there would be no point in experimenting
with a vaccine for which the expected VE is negative, since
that case would imply $\lambda_v>\lambda_c$).

Finally, the choice of $a_v$ could be driven by noticing
that $a_v=1$ allows for an inverted J-shape density on
$\lambda_v$, which is therefore exponential.
WIth this choice, $b_c=b_v=(\widehat{\lambda_c})^{-1}$, i.e. the average
time to infection of a randomly selected participant in the control group.
Notice also that in this case $a_c>1$, which would allow for the
prior expectation of VE in formula (\ref{priorVE}) to exist,
positive.

To recap, here's the full Bayesian model proposed, which is also illustrated
graphically by its associated Directed acyclic Graph (DAG) in Figure~\ref{figura}
\begin{align*}
  \lambda_v &\sim \text{Gamma}(a_v,b_v) \\
  \lambda_c &\sim \text{Gamma}(a_c,b_c) \\
  s_v|\lambda_v &\sim \text{Normal}(n_v \mathrm{E}(\min(T_v,C_v)), n_v \mathrm{Var}(\min(T_v,C_v)) \\
  s_c|\lambda_c &\sim \text{Normal}(n_c \mathrm{E}(\min(T_c,C_c)), n_c \mathrm{Var}(\min(T_c,C_c)) \\
  X_V+X_C|S_V,S_C,\lambda_V,\lambda_C &\sim \text{Poisson}(s_V\lambda_V+s_C\lambda_C) \\
X_V|X_V+X_C,S_V,S_C,\lambda_v,\lambda_c &\sim \text{Binomial} \left(x_v+x_c,
  \frac{s_v\lambda_v}{s_v\lambda_v+s_c\lambda_c} \right).
\end{align*}
with $\mathrm{E}(\min(T_v,C_v)) \mathrm{Var}(\min(T_v,C_v)),
\mathrm{E}(\min(T_c,C_c)) \mathrm{Var}(\min(T_c,C_c))$ given in Theorem 3.1
for the vaccine and the control group respectively and the following
default choices:
\begin{align*}
  a_v &=1 \\
  a_c &=\frac{2-\widehat{VE}}{1-\widehat{VE}} \\
  b_v=b_c &= (\widehat{\lambda_c})^{-1}.
\end{align*}
The following theorem draws a connection to
the exact method conditional on the total number of cases.
\begin{theorem}
  If $\lambda_v \sim \mathrm{Gamma}(a_v,b_v)$ and, independently,
  \linebreak $\lambda_c \sim \mathrm{Gamma}(a_c,b_c)$, then
  $$ \frac{b_v\lambda_v}{b_v\lambda_v+b_c\lambda_c}
  \sim \mathrm{Beta}(a_v,a_c).$$
\end{theorem}
\begin{proof}
  It is easy to see that $b_v\lambda_v \sim  \text{Gamma}(a_v,1)$
  and  , independently, $b_c\lambda_c \sim  \text{Gamma}(a_c,1)$.
  Then, by the  well known Renyi's representation of a Dirichlet
  distribution,   which in one dimension is the same as a Beta distribution,
  the theorem follows.
\end{proof}
If we used the data-dependent prior $b_v=s_v$ and $b_c=s_c$,
then the full model would approximate  the
exact method conditional on the total number of
cases (see equation~\ref{relation}).
That is not what is recommended here though, since data-dependent
priors are difficult to accept and instead the prior choice
$b_v=b_c = (\widehat{\lambda_c})^{-1}$ discussed above looks more reasonable.


\section{The Pfizer/BioNTech protocol revisited}

Several computational strategies are available for the full Bayesian
model, the easiest being MCMC simulation of the exact posterior
distribution of VE.
Among the many possibilities available nowadays, the sofware
OpenBUGS \linebreak
(\texttt{http://www.openbugs.net/}) has been used here.
The model and data files necessary to run the MCMC simulation
in OpenBUGS are listed in the Appendix.
The following results are taken from Table 2 of \cite{Pfizer}:
\begin{align*}
  n_c &= 17411 &\text{sample size of the vaccinated cohort}\\
  n_v &= 17511 &\text{sample size of the control cohort}\\
  s_v &= 2214 &\text{surveillance time of the vaccinated cohort}\\
  s_c &= 2222 &\text{surveillance time of the control cohort}\\
  x_v &= 8 &\text{number of cases in the vaccinated cohort}\\
  x_c & =162 &\text{number of cases in the control cohort}\\
  D&=0.29 &\text{enrollment duration, in years}
\end{align*}
and, according to the discussion at the end of Section 3,
they have been programmed together with the following
choice of the parameters to mimick the method used
in \cite{Pfizer}:
\begin{align*}
      a_v=0.7, \ \ b_v=2214 \ \ a_c=1, \ \ b_c=2222.
\end{align*}
The results of the Bayesian analysis are summarized in the
posterior mean of percentage VE, equal to 93.7,
and in the posterior equal tail 95\% credible interval (89.0,97.0).
The results reported in \cite{Pfizer} are instead a posterior mean
equal to 95.0 and a posterior interval (90.3,97.6) for VE. 
We notice a basic agreement of the results, but a larger uncertainty
regarding VE in the full Bayesian
model due to the more accurate complete accounting of sampling variation.

By proceeding instead following the recommendations at the end
of Section 3 and assuming an average infection time equal to a week
($(\widehat{\lambda_c})^{-1}=0.01917808$ years), one would have used instead 
\begin{align*}
      a_v=1, \ \ b_v=0.01917808 \ \ a_c=2.428571, \ \ b_c=0.01917808.
\end{align*}
and obtained a posterior mean
equal to 93.6 and a posterior interval (89.0,97.0) for VE.
That does not differ from the full Bayes analysis due to the very large
sample sizes, which make the influence of the prior disappear and
provide substantial evidence for the validity of the Pfizer/BioNTech vaccine.

\section{Conclusions}

From a theoretical point of view, a new fully Bayesian coherent model is
developed here for VE. It does not provide results in strong contrast
with the approximations in the original paper \cite{Pfizer}, which exhibits
large sample sizes and uncontroversial positive results. The new model may be the
theoretical basis for the other vaccines currently under development,
which may not exhibit the same size of VE. For those, a careful and motivated
prior may make a difference.

\begin{figure}
\setlength{\unitlength}{.25cm}
\begin{center}
\begin{picture}(40,35)(0,15)

\put(20,29){\oval(5,2)}
\put(20,29){\makebox(0,0){$\lambda_v$}}
\put(20,28){\vector(0,-1){2}}
\put(19,28){\vector(-2,-1){10}}
\put(22.5,29){\vector(1,-2){3}}
\put(22.5,28.8){\line(1,-2){2.8}}

\put(20,25){\oval(5,2)}
\put(20,25){\makebox(0,0){$S_v$}}
\put(19,24){\vector(0,-1){1}}
\put(17.5,25){\vector(-3,-1){6}}

\put(17,22){\oval(7,2)}
\put(17,22){\makebox(0,0){$X_v+X_c$}}
\put(13.5,22){\vector(-1,0){1}}

\put(20,19){\oval(5,2)}
\put(20,19){\makebox(0,0){$S_c$}}
\put(19,20){\vector(0,1){1}}
\put(17.5,19){\vector(-3,1){6}}

\put(20,15){\oval(5,2)}
\put(20,15){\makebox(0,0){$\lambda_c$}}
\put(20,16){\vector(0,1){2}}
\put(19,16){\vector(-2,1){10}}
\put(22.5,15){\vector(1,2){3}}
\put(22.5,15.2){\line(1,2){2.8}}

\put(10,22){\oval(5,2)}
\put(10,22){\makebox(0,0){$X_v$}}

\put(26,22){\oval(4,2)}
\put(26,22){\makebox(0,0){VE}}

\end{picture}   
\end{center}
\caption{DAG for the full Bayesian model
  for Vaccine Efficacy. Double arrows indicate a functional relationship.}
  \label{figura}
\end{figure}
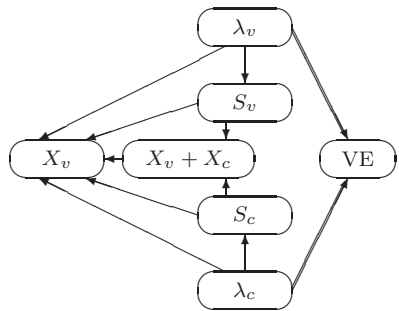

\section{Appendix}

The model of the OpenBUGS program is
\begin{verbatim}
model{

lambdac ~ dgamma(ac,bc)
lambdav ~ dgamma(av,bv)
sc ~ dnorm(expsc, tausc)
sv ~ dnorm(expsv, tausv)

tausc <- 1/(nv*varsc)
tausv <- 1/(nv*varsv)

expsc <- nc*(1- (1-exp(-lambdac*D))/(lambdac*D)) / lambdac
varsc <- (2*exp(-lambdac*D) + 4*exp(-lambdac*D)/(lambdac*D) - 
          pow((1-(1-exp(-lambdac*D))/(lambdac*D)),2))/pow(lambdac,2)
expsv <- nv*(1- (1-exp(-lambdav*D))/(lambdav*D))/lambdav
varsv <- (2*exp(-lambdav*D) + 4*exp(-lambdav*D)/(lambdav*D) -
         pow((1-(1-exp(-lambdav*D))/(lambdav*D)),2))/pow(lambdav,2)

cases ~ dpois(meancases)
meancases <- sc*lambdac + sv*lambdav

xv ~ dbin(theta, cases)
theta <- sv*lambdav/(sc*lambdac + sv*lambdav)

VE <- 1 - lambdav/lambdac
}
\end{verbatim}
whereas the data of the OpenBUGS program is
\begin{verbatim}
  list(av=1, bc=0.01917808,         ### full Bayes prior for lambdac 
       ac= 2.428571, bv=0.01917808,   ### full Bayes prior for lambdav
#      ac=1, bc=2222,     ### Pfizer prior for lambdac 
#      av=0.7, bv=2214,   ### Pfizer prior for lambdav
      nc=17511,       ### ss for c
      nv=17411,       ### ss for v
      sc=2222,        ### surveillance time for c
      sv=2214,        ### surveillance time for v
      cases=170,      ### totalnumber of cases
      xv=8,           ### infected among vaccinated
      D=0.29          ### duration of uniform recruitment
)

\end{verbatim}

\end{document}